\pdfoutput=1
\documentclass[10pt,letterpaper,aps,prd,notitlepage,longbibliography,nofootinbib]{revtex4-1}
\usepackage[english]{babel}
\usepackage[left=0.94in,right=0.94in,top=0.94in,bottom=0.94in]{geometry}
\usepackage{microtype}
\usepackage[utf8]{inputenc}
\usepackage{amsmath}
\usepackage{amsfonts}
\usepackage{amssymb}
\usepackage{amsthm}
\usepackage{braket}
\usepackage{bbm}
\usepackage[pdfusetitle]{hyperref}
\usepackage[all]{hypcap}
\usepackage{booktabs}
\usepackage{tensor}
\usepackage{xcolor}
\hypersetup{
	colorlinks,
	linkcolor={red!80!black},
	citecolor={blue!70!black},
	urlcolor={purple!90!black}
}
\renewcommand{\vec}{\mathbf}

\DeclareMathOperator{\Span}{span}
\DeclareMathOperator{\sgn}{sgn}
\DeclareMathOperator{\Lin}{Lin}
\newcommand{\ii}{\mathbbm{i}}
\newcommand{\half}[1][1]{\tfrac{#1}{2}}
\newcommand{\shalf}[1][1]{\frac{#1}{2}} 
\newcommand{\Rho}{\mathrm{P}}
\newcommand{\casesif}{\quad\textnormal{if }\,} 
\newcommand{\casestextn}[1]{\quad\textnormal{#1}} 
\newcommand{\setst}{\,|\,} 
\newcommand{\CG}{Clebsch--Gordan}
\newcommand{\JS}{Jordan--Schwinger}
\newcommand{\WE}{Wigner--Eckart}
\newcommand{\R}{\mathbb{R}}
\newcommand{\C}{\mathbb{C}}
\newcommand{\K}{\mathbb{K}}
\newcommand{\Z}{\mathbb{Z}}
\newcommand{\N}{\mathbb{N}}
\newcommand{\1}{\mathbbm{1}}
\newcommand{\cC}{\mathcal{C}}
\newcommand{\cJ}{\mathcal{J}}
\newcommand{\cM}{\mathcal{M}}
\newcommand{\GL}{\mathrm{GL}}
\newcommand{\SL}{\mathrm{SL}}
\newcommand{\SO}{\mathrm{SO}}
\newcommand{\SU}{\mathrm{SU}}
\newcommand{\Spin}{\mathrm{Spin}}

\newcommand{\su}{\mathfrak{su}}
\newcommand{\spin}{\mathfrak{spin}}

\newcommand*\widebar[1]{%
  \hbox{
    \vbox{%
      \hrule height 0.2pt 
      \kern0.35ex
      \hbox{%
        \kern-0.1em
        \ensuremath{#1}%
        \kern-0.1em}}}}
\theoremstyle{plain}
\newtheorem{proposition}{Proposition}
\newtheorem{definition}{Definition}
\newtheorem*{theorem*}{Theorem}
\newtheorem*{remark}{Remark}
\newtheorem{corollary}{Corollary}
\newtheorem{lemma}{Lemma}
\newtheorem{appxprop}{Proposition}

\begin{document}
\title{\WE\ theorem and \JS\ representation for infinite-dimensional representations of the Lorentz group}
\author{Giuseppe Sellaroli}\email{gsellaroli@uwaterloo.ca}
\affiliation{Department of Applied Mathematics, University of Waterloo, Waterloo\\Ontario N2L 3G1, Canada}
\date{\today}
\begin{abstract}
The Wigner--Eckart theorem is a well known result for tensor operators of $\SU(2)$ and, more generally, any compact Lie group.
This paper generalises it to arbitrary Lie groups, possibly non-compact.
The result relies on knowledge of recoupling theory between finite-dimensional and arbitrary admissible representations, which may be infinite-dimensional; the particular case of the Lorentz group will be studied in detail.
As an application, the Wigner--Eckart theorem will be used to construct an analogue of the Jordan--Schwinger representation, previously known only for finite-dimensional representations of the Lorentz group, valid for infinite-dimensional ones.
\end{abstract}
\maketitle
\section{Introduction}
Representation theory of Lie groups has many applications in modern physics, especially in quantum theory.
In particular tensor operators for compact Lie groups have been widely used in non-relativistic quantum mechanics~\cite{messiah2}, mainly because of the fact---known as \emph{\WE\ theorem}---that their matrix elements can be expressed as a product of a \CG\ coefficient\footnote{The coefficients
appearing in the decomposition of the tensor product of two irreducible representations as the direct sum of irreducible representations.} and a factor independent of the particular basis vectors being evaluated.

Although only compact groups are generally considered, one may wonder if these result extends to non-compact ones as well.
A generalisation of the \WE\ theorem to non-compact groups has been considered in~\cite{locallycompact}, but only for tensor operators ``transforming''\footnote{See Section \ref{subsec:tensor_operators} for the precise definition of this statement.} as unitary representations, which are generally infinite-dimensional in the non-compact case; here tensor operators with finitely many components will be considered.
The particular case of $\SL(2,\R)$ was previously studied by the author in~\cite{wigner_eckart}; in the present paper that result is generalised to any Lie group. The information provided by the \WE\ theorem, however, relies on the knowledge of the \CG\ decomposition of the product of a finite-dimensional representation and an arbitrary admissible one, which may be infinite-dimensional if the group is non-compact. Since a general result about this decomposition is not available, each case has to be treated individually; in particular the couplings for the Lorentz group---previously unconsidered\footnote{To the best of the author's knowledge.}---will be considered here. Most of the paper will be dedicated on this task.

An application of the theorem will also be considered.
An important result of $\SU(2)$ representation theory, especially useful in quantum field theory, is the \JS\ representation, i.e., the reformulation of the $\su(2)$ generators in terms of two quantum harmonic oscillator operators.
A similar result for the Lorentz group exists only for finite-dimensional representations, which can be seen as the product of two $\SU(2)$ representations (see Section~\ref{sec:representation_theory}).
It will be shown that, making use of the \WE\ theorem, the infinite-dimensional representations admit an analogous construction in terms of tensor operators.

The paper is organized as follows: Section~\ref{sec:representation_theory} contains a review of representation theory of the Lorentz group and the analysis of the product of a finite-dimensional and an infinite-dimensional representation.
Results concerning tensor operators are presented in Section~\ref{sec:tensor_operators}: definition, \WE\ theorem and \JS\ representation.
Finally, a table of notations used throughout the paper, a table of \CG\ coefficients for the coupling of finite and infinite-dimensional representations and some results needed in Section~\ref{sec:representation_theory} are included as appendices.
\section{Representation theory}\label{sec:representation_theory}
This section contains both preliminary notions and new results about the representation theory of the Lorentz group. First the definition of Lorentz group and Lorentz algebra and their irreducible admissible representations will be recalled, then the decomposition of the product of a finite-dimensional and an infinite-dimensional representation will be studied. The reviews of representation theory follows~\cite{Harish-Chandra1947,knapp,wallach}.
\subsection{The Lorentz group and its Lie algebra}

The \emph{proper orthochronous Lorentz group} $\SO_0(3,1)$, henceforth simply referred to as the Lorentz group, is the identity component of the subgroup of $\GL(4,\R)$ that preserves the quadratic form
\begin{equation}
Q(x)=-(x_0)^2 + (x_1)^2 + (x_2)^2 + (x_3)^2,\quad x=(x_0,x_1,x_2,x_3)\in \R^4.
\end{equation}
To allow for spin representations, the double cover $\Spin(3,1)\cong \SL(2,\C)_\R$ of $\SO_0(3,1)$ will be used here; moreover, only complex representations will be considered, so that one may work with  a complexified Lie algebra.

The Lie algebra $\spin(3,1)_\C$ has 6 generators
\begin{equation}
\vec{J}=(J_0,J_1,J_2),\quad \vec{K}=(K_0,K_1,K_2),
\end{equation}
with commutation relations
\begin{equation}\label{eq:spin(3,1)_commutation}
[J_a,J_b]=\ii\tensor{\varepsilon}{_{ab}^c}J_c,\quad [J_a,K_b]=\ii\tensor{\varepsilon}{_{ab}^c}K_c,\quad[K_a,K_b]=-\ii\tensor{\varepsilon}{_{ab}^c}J_c.
\end{equation}
The $J$'s generate the subalgebra $\spin(3)\cong\su(2)$ (i.e., spatial rotations), while the $K$'s are the generators of boosts.
The algebra has two Casimirs
\begin{equation}
\cC_1=\vec{J}\cdot\vec{K},\quad \cC_2=J^2-K^2
\end{equation}
which, introducing the operators
\begin{equation}
J_\pm:=J_1\pm i J_2,\quad K_\pm:=K_1\pm i K_2
\end{equation}
and making use of~\eqref{eq:spin(3,1)_commutation},
can be rewritten as
\begin{equation}
\cC_1=J_0K_0+\tfrac{1}{2}(J_-K_+ + J_+K_-),\quad \cC_2=J^2-(J_0 + K_0^2+K_+K_-).
\end{equation}

\subsection{Irreducible representations of the Lorentz group}
We will restrict to \emph{admissible} Hilbert space representations (which include, in particular, every unitary representation), i.e., those that are unitary when restricted to the maximal compact subgroup $\SU(2)$ and such that each irreducible unitary representation of $\SU(2)$ appears in it with finite multiplicity.

Each admissible representation induces an admissible $(\mathfrak{g},K)$-module\footnote{An algebraic object that, given a Lie group $G$ with Lie algebra $\mathfrak{g}$ and a maximal compact subgroup $K$, is both a Lie algebra representation of $\mathfrak{g}$ and a group representation of $K$, with the appropriate compatibility conditions ~\cite{wallach}. Here a \emph{module} is a vector space equipped with a linear action of one of more rings on it (e.g., $\mathfrak{g}$ and $K$).}, with $\mathfrak{g}=\spin(3,1)$ and $K=SU(2)$, which as an algebraic object is easier to work with; although not true in general, for $\Spin(3,1)$ admissible irreducible representations of the group and admissible irreducible $(\mathfrak{g},K)$-modules are in one-to-one correspondence\footnote{Compare the list of admissible group representations in~\cite[chap. VIII]{knapp} and the list of admissible $(\mathfrak{g},K)$-modules in~\cite[chap. 4]{wallach}.}.

The general irreducible admissible $(\mathfrak{g},K)$-module, labelled by a pair $(\lambda,\rho)\in\Z/2\times\C$, is the \emph{algebraic} direct sum\footnote{i.e., only sums of finitely many vectors are considered.}
\begin{equation}
V_{\lambda,\rho} =\bigoplus_{j=|\lambda|}^{j_\textnormal{max}} V^j_{\lambda,\rho}
\end{equation}
of unitary irreducible $\SU(2)$-modules $V^j_{\lambda,\rho}$, where the sum is in integer steps and, depending on the values of $\lambda$ and $\rho$, it is either $j_\textnormal{max}\in |\lambda|+\N_0$ or $j_\textnormal{max}=\infty$ (see later discussion). The (complex) vector space $V_{\lambda,\rho}^j$ is spanned by the basis
\begin{equation}
\label{eq:basis}
\ket{(\lambda,\rho)\,j,m}, \quad m\in\cM_j:=\lbrace-j,-j+1,\dotsc,j-1,j\rbrace,
\end{equation}
on which the $\su(2)_\C$ generators act as
\begin{equation}
\begin{cases}
J_0\ket{(\lambda,\rho)\,j,m}=m\ket{(\lambda,\rho)\,j,m}\\
J_\pm \ket{(\lambda,\rho)\,j,m}=C_\pm(j,m)\ket{(\lambda,\rho)\,j,m\pm 1}\\
J^2\ket{(\lambda,\rho)\,j,m}=j(j+1)\ket{(\lambda,\rho)\,j,m},
\end{cases}
\end{equation}
with
\begin{equation}
C_\pm(j,m):=\sqrt{j\mp m}\sqrt{j \pm m+1},
\end{equation}
i.e., they are eigenvectors for $J_0$ and $J^2$; since $\SU(2)$ is simply connected, its action on $V_{\lambda,\rho}^j$ is completely determined by the corresponding Lie algebra action.
The $(\mathfrak{g},K)$-module will be given an inner product by requiring the $\SU(2)$-modules to be orthogonal to each other and the vectors in~\eqref{eq:basis} to be orthonormal.

The possible matrix elements of the boost generators are
\begin{subequations}
\label{eq:K_actions}
\begin{align}
\braket{j+1,m\pm1|K_\pm|j,m}&=\mp P^+_{\lambda,\rho}(j)\sqrt{j\pm m +1}\sqrt{j\pm m +2}\\
\braket{j+1,m|K_0|j,m}&= P^+_{\lambda,\rho}(j)\sqrt{j+ m +1}\sqrt{j- m +1}\\[1em]
\braket{j,m\pm1|K_\pm|j,m}&= P_{\lambda,\rho}(j)\,C_\pm(j,m)\\
\braket{j,m|K_0|j,m}&= P_{\lambda,\rho}(j)\,m\\[1em]
\braket{j-1,m\pm1|K_\pm|j,m}&=\pm P^-_{\lambda,\rho}(j)\sqrt{j\mp m}\sqrt{j\mp m -1}\\
\braket{j-1,m|K_0|j,m}&= P^-_{\lambda,\rho}(j)\sqrt{j+ m}\sqrt{j- m},
\end{align}
\end{subequations}
where
\begin{equation}
P^-_{\lambda,\rho}(j)=\frac{\sqrt{j+\lambda}\sqrt{j-\lambda}\sqrt{j+\rho}\sqrt{j-\rho}}{j\sqrt{2j+1}\sqrt{2j-1}},\quad P_{\lambda,\rho}^+(j)=P_{\lambda,\rho}^-(j+1)
\end{equation}
and
\begin{equation}
P_{\lambda,\rho}(j)=
\begin{cases}
\frac{\ii\lambda\rho}{j(j+1)}&\casesif j\neq 0\\
0&\casesif j=0.
\end{cases}
\end{equation}
The Casimirs act on $V_{\lambda,\rho}$ as
\begin{equation}
\cC_1=\ii\lambda\rho\,\1,\quad \cC_2=(\lambda^2+\rho^2-1)\1.
\end{equation}
The values $j$ can take have an upper bound $j_\textnormal{max}\in |\lambda|+\N_0$ if and only if
\begin{equation}
P_{\lambda,\rho}^+(j_\textnormal{max})=0\quad\textnormal{and}\quad P_{\lambda,\rho}^+(j)\neq 0 \quad\forall j<j_\textnormal{max},
\end{equation}
i.e.,
\begin{equation}
\rho=\pm(j_\textnormal{max}+1).
\end{equation}
It follows that $V_{\lambda,\rho}$ is finite-dimensional when $\rho\in \pm (|\lambda|+\N)$ and it is infinite-dimensional in all other cases.
\begin{remark}[isomorphic modules]
The values of the Casimirs and of $P_{\lambda,\rho}(j)$, $P^+_{\lambda,\rho}(j)$ and $P^-_{\lambda,\rho}(j)$ are invariant under the change $(\lambda,\rho)\rightarrow (-\lambda,-\rho)$;
moreover, whether the module is finite-dimensional and the eventual value of $j_\textnormal{max}$ are unaffected by the change as well. It follows that the modules $V_{\lambda,\rho}$ and $V_{-\lambda,-\rho}$ are isomorphic.
\end{remark}

Unitary modules are those for which
\begin{equation}
K_0^\dagger=K_0,\quad K_+^\dagger=K_-,
\end{equation}
with respect to the inner product on $V_{\lambda,\rho}$.
Explicitly, it must be
\begin{equation}
\widebar{P_{\lambda,\rho}(j)}=P_{\lambda,\rho}(j),\quad \widebar{P^-_{\lambda,\rho}(j)}=P^-_{\lambda,\rho}(j),
\end{equation}
which is satisfied by three possible classes of modules:
\begin{itemize}
\item \emph{principal series}: $\lambda\in\Z/2$ and $\rho \in \ii\R$;
\item \emph{complementary series}: $\lambda=0$ and $\rho\in(-1,0)\cup(0,1)$;
\item \emph{trivial representation}: $\lambda=0$ and $\rho=\pm1$.
\end{itemize}

\subsubsection*{Finite-dimensional modules}\label{par:finite-dimensional}
It was shown that the $(\mathfrak{g},K)$-module with $\rho=B(\omega+1)$, $\omega\in |\lambda|+\N_0$, $B=\pm 1$ is finite-dimensional. We will assume, for finite-dimensional modules (and for those only), that $\lambda\geq 0$. It is then easy to check that
\begin{equation}
\dim V_{\lambda,\rho} =
\sum_{j=\lambda}^\omega (2j+1)=(\omega-\lambda+1)(\omega+\lambda+1).
\end{equation}
Finite-dimensional modules can be given an alternative construction using the fact that
\begin{equation}
\spin(3,1)_\C\cong \su(2)_\C \oplus \su(2)_\C,
\end{equation}
i.e., by changing to the basis
\begin{equation}
\vec{M}^A:=\half(\vec{J}-\ii A \vec{K}),\quad A=\pm 1,
\end{equation}
with commutation relations
\begin{equation}
[M^A_a,M^B_b]=\ii \tensor{\varepsilon}{_{ab}^c}M^A_c\delta_{AB};
\end{equation}
one can easily show that, for finite-dimensional modules,
\begin{equation}
K_0^\dagger=-K_0,\quad K_+^\dagger = - K_-,
\end{equation}
so that each $M_a^A$ is self-adjoint, i.e., the action of each $\su(2)$ subalgebra is unitary.

From $\su(2)$ representation theory we know that, if $V_j$ is the $(2j+1)$-dimensional unitary irreducible $\su(2)$-module,
\begin{equation}
\bigoplus_{j=\lambda}^\omega V_j \cong V_{\half[\omega+\lambda]}\otimes V_{\half[\omega-\lambda]}\cong V_{\half[\omega-\lambda]}\otimes V_{\half[\omega+\lambda]}.
\end{equation}
Since $\vec{J}=\vec{M}^{(-)} + \vec{M}^{(+)}$, one can then change the basis to
\begin{equation}
\ket{j_1,m_1}\otimes \ket{j_2,m_2}=\sum_{j=\lambda}^\omega \sum_{m=-j}^j \braket{j,m|j_1,m_1;j_2,m_2}\ket{(\lambda,\rho)\,j,m},
\end{equation}
where $\braket{j,m|j_1,m_1;j_2,m_2}$ are the $\su(2)$ \emph{\CG\ coefficients}\footnote{The Condon–Shortley convention~\cite{condon_shortley} is used here.} and
\begin{equation}
\begin{cases}
j_1=\frac{\omega+\lambda}{2}\\
j_2=\frac{\omega-\lambda}{2}
\end{cases}
\textnormal{or}\quad
\begin{cases}
j_1=\frac{\omega-\lambda}{2}\\
j_2=\frac{\omega+\lambda}{2};
\end{cases}
\end{equation}
it is assumed that $\vec{M}^{(-)}$ and $\vec{M}^{(+)}$ only act respectively on $\ket{j_1,m_1}$ and $\ket{j_2,m_2}$. The dimension of the new basis is
\begin{equation}
(2j_1+1)(2j_2+1)=(\omega+\lambda+1)(\omega-\lambda+1),
\end{equation}
as expected.
The choice of $j_1$, $j_2$ depends on the sign of $\rho$: in fact, one has
\begin{equation}
\label{eq:finite-dim_consistency}
\cC_1\ket{j_1,m_1}\otimes \ket{j_2,m_2}=\ii B \lambda(\omega+1)\ket{j_1,m_1}\otimes \ket{j_2,m_2},
\end{equation}
but also
\begin{equation}
\cC_1=\sum_A \ii A \left(M^A\right)^2,
\end{equation}
so that~\eqref{eq:finite-dim_consistency} is consistent if and only if
\begin{equation}
j_1=\frac{\omega-B\lambda}{2},\quad j_2=\frac{\omega+B\lambda}{2}.
\end{equation}

Conversely, one can show that every product of $\su(2)$-modules
\begin{equation}
V_{j_1}\otimes V_{j_2},\quad j_1,j_2\in\N_0/2
\end{equation}
gives rise to a Lorentz $(\mathfrak{g},K)$-module with
\begin{equation}
\lambda=|j_1-j_2|,\quad \rho=
\begin{cases}
(j_1+j_2+1)&\casesif j_1 < j_2\\
-(j_1+j_2+1) & \casesif j_1\geq j_2.
\end{cases}
\end{equation}
As a consequence, every finite-dimensional irreducible $(\mathfrak{g},K)$-module can be specified by a pair $(j_1,j_2)\in \N_0/2 \times \N_0/2$; is it customary to use the pair to denote the module itself.

Examples of finite-dimensional modules are
\begin{itemize}
\item $(0,0)$: the \emph{scalar} module (trivial representation);
\item $(\half,0)$ and $(0,\half)$: respectively the left and right \emph{Weyl spinor} modules;
\item $(\half,\half)$: the (complexified) \emph{vector} module;
\item $(\half,0)\oplus(0,\half)$: the \emph{Dirac spinor} module (not irreducible).
\end{itemize}
It is not difficult to infer the decomposition of the product of two finite-dimensional modules from $\su(2)$ results; one has
\begin{equation}
(j_1,k_1)\otimes(j_2,k_2)\cong \bigoplus_{j=|j_1-j_2|}^{j_1+j_2}\bigoplus_{k=|k_1-k_2|}^{k_1+k_2}(j,k),
\end{equation}
and, in particular,
\begin{equation}
\label{eq:finite_recoup}
(j_1,j_2)\equiv (j_1,0)\otimes(0,j_2).
\end{equation}
We will refer to modules of the kind $(j,0)$ and $(0,j)$ respectively as \emph{left} and \emph{right} modules; it follows from~\eqref{eq:finite_recoup} that any other irreducible module can be constructed from the product of a left and right one. To allow to easily specify if a module is left or right, the notation
\begin{equation}
F^A_j:=
\begin{cases}
(j,0) & \casesif A=-1\\
(0,j) &\casesif A=1
\end{cases}
\end{equation}
will be used in the following sections. A basis for $F^A_j$ is given by
\begin{equation}
\ket{j_A,\mu},\quad \mu\in\cM_j,
\end{equation}
with
\begin{equation}
\begin{cases}
J_0\ket{j_A,\mu}=\mu \ket{j_A,\mu}\\
J_\pm \ket{j_A,\mu}=C_\pm(j,\mu)\ket{j_A,\mu\pm 1}
\end{cases}
\textnormal{and}\quad
\begin{cases}
K_0\ket{j_A,\mu}=\ii A \mu \ket{j_A,\mu}\\
K_\pm \ket{j_A,\mu}=\ii A C_\pm(j,\mu)\ket{j_A,\mu\pm 1}.
\end{cases}
\end{equation}

\subsection{Decomposition of the product of finite and infinite-dimensional modules}\label{sec:CG_decomposition}

We are interested in the tensor product of a non-trivial finite-dimensional module (necessarily non-unitary) and an infinite-dimensional one (either unitary or non-unitary).
In light of the consequences of~\eqref{eq:finite_recoup} mentioned above, it will suffice to restrict ourselves to couplings of the kind $F^A_\gamma \otimes V_{\lambda,\rho}$, where $\gamma\geq \half$ and $\lambda$, $\rho$ are such that 
\begin{equation}
P^+_{\lambda,\rho}(j)\neq 0,\quad \forall j \in |\lambda|+\N_0.
\end{equation}
The goal is to study the \emph{\CG\ decomposition} of such a product, i.e., wheter it is possible to rewrite such a module (generally reducible) as a sum of irreducible ones.
Equivalently, one can ask if it is possible to \emph{simultaneously} diagonalise the two Casimirs in the product module\footnote{One can check explicitly that the Casimirs acting on the product space are neither self-adjoint nor normal operators, i.e., $[\mathcal{C}_a,\mathcal{C}_a^\dagger]\neq 0$, so that the spectral theorem cannot be used.}, where the generators act as
\begin{equation}
\vec{J}\equiv\vec{J}\otimes \1 + \1\otimes \vec{J},\quad \vec{K}\equiv\vec{K}\otimes \1 + \1\otimes \vec{K}
\end{equation}
on the basis elements
\begin{equation}
\label{eq:product_basis}
\ket{\gamma_A,\mu}\otimes\ket{(\lambda,\rho)\,j,m},\quad j\in |\lambda|+\N_0,\quad m \in \cM_j,\quad \mu \in \cM_\gamma.
\end{equation}

Instead of working with an infinite dimensional vector space, one can decompose the product space into a sum of finite-dimensional spaces: this is accomplished by diagonalising $J_0$ and $J^2$ first. Using $\su(2)$ recoupling theory, one finds that the vectors
\begin{equation}
\ket{(j)\,J,M}:=\sum_{\mu\in\cM_\gamma}\sum_{m\in\cM_j} \braket{\gamma,\mu;j,m|J,M}\ket{\gamma_A,\mu}\otimes\ket{(\lambda,\rho)\,j,m},\quad j\in |\lambda|+\N_0,\quad J\in\lbrace |j-\gamma|,\dotsc,j+\gamma \rbrace
\end{equation}
provide an \emph{orthonormal} basis of $(J_0,J^2)$-eigenvectors for the product space.
The eigenspace $V^J_M$, defined by
\begin{equation}
J_0\ket{\psi}=M\ket{\psi},\quad J^2\ket{\psi}=J(J+1)\ket{\psi},\quad \forall\ket{\psi}\in V^J_M,
\end{equation}
is spanned by the basis vectors
\begin{equation}
\ket{(j)\,J,M},\quad j\in\Omega_J(\lambda,\gamma):=
\begin{cases}
\lbrace \max(|\lambda|,J-\gamma),\dotsc,J+\gamma\rbrace &\casesif J\geq ||\lambda|-\gamma |\\
\lbrace \gamma-J,\dotsc,\gamma+J\rbrace & \casesif J< ||\lambda|-\gamma |,
\end{cases}
\end{equation}
so that
\begin{equation}
\dim V^J_M =
\begin{cases}
\min(J+\gamma-|\lambda|+1,2\gamma+1) &\casesif J\geq ||\lambda|-\gamma |\\
2J+1&\casesif J< ||\lambda|-\gamma |;
\end{cases}
\end{equation}
note that, when $|\lambda|\geq \gamma$, it is always true that $J\geq ||\lambda|-\gamma |=|\lambda|-\gamma$, so that the case $J< ||\lambda|-\gamma |$ only needs to be considered when $|\lambda|<\gamma$. The set of possible values of $J$ is
\begin{equation}
\mathcal{J}(\lambda,\gamma):=\max(\varepsilon,|\lambda|-\gamma)+\N_0,\quad \varepsilon=
\begin{cases}
0 &\casesif \lambda+\gamma\in \Z\\
\half  &\casesif \lambda+\gamma\in \half+\Z.
\end{cases}
\end{equation}
It will prove useful in the following to build $\mathcal{J}(\lambda,\gamma)$ in a different way.
First notice that
\begin{equation}
|\lambda+\nu|\in\mathcal{J}(\lambda,\gamma),\quad \forall\nu\in\cM_\gamma,
\end{equation}
as
\begin{equation}
\begin{cases}
|\lambda+\nu|=|\lambda|+\sgn(\lambda)\nu\in |\lambda|-\gamma +\N_0&\casesif|\lambda|\geq\gamma \\
|\lambda+\nu|\in \varepsilon+\N_0&\casesif|\lambda|<\gamma.
\end{cases}
\end{equation}
Let then
\begin{equation}
\Sigma_\nu(\lambda,\gamma):=|\lambda+\nu|+\N_0,\quad \nu\in \cM_\gamma;
\end{equation}
simple counting arguments show that
\begin{equation}
\label{eq:sigma_decomposition}
\bigcup_{\nu\in\cM_\gamma}\Sigma_\nu(\lambda,\gamma)=\mathcal{J}(\lambda,\gamma)\quad\textnormal{and}\quad \sum_{\nu\in\cM_\gamma}\chi_{\Sigma_\nu}(J)=\dim V^J_M,
\end{equation}
where $\chi$ is the \emph{indicator function}
\begin{equation}
\chi_A(x):=
\begin{cases}
1&\casesif x\in A\\
0&\casesif x\not\in A,
\end{cases}.
\end{equation}

Since the Casimirs commute with both $J_0$ and $J^2$, one can work with their restriction on the finite-dimensional subspaces $V^J_M$ and diagonalise those; moreover, it suffices to consider the restrictions to $V_J:=V^J_J$ thanks to the following
\begin{proposition}\label{prop:JM}
Let $J\in\mathcal{J}(\lambda,\gamma)$. The eigenvalues of the Casimirs $\cC_1$ and $\cC_2$ are the same on each $V^J_M$, $M\in\cM_J$.
\end{proposition}
\begin{proof}
The basis vectors of $V^J_M$ satisfy
\begin{equation}
J_\pm\ket{(j)\,J,M}=C_\pm(J,M)\ket{(j)\,J,M\pm 1},
\end{equation}
so that
\begin{equation}
J_\pm (V^J_M) \subseteq V^J_{M\pm 1},\quad
\begin{cases}
\ker J_+|_{V^J_M}=\lbrace \vec{0} \rbrace, \quad &\forall M<J\\
\ker J_-|_{V^J_M}=\lbrace \vec{0} \rbrace, \quad &\forall M>-J.
\end{cases}
\end{equation}
Since $J_\pm$ commutes with the Casimirs, given a $\cC_a$-eigenvector $\ket{\alpha_a}\in V^J_M$ with eigenvalue $\alpha_a\in \C$ one has
\begin{equation}
0\neq J_\pm\ket{\alpha_a}\in V^J_{M\pm 1},\quad \cC_a J_\pm \ket{\alpha_a}=J_\pm \cC_a\ket{\alpha_a}=\alpha_a J_\pm\ket{\alpha_a}
\end{equation}
whenever $V^J_{M\pm 1}$ is defined, so that each $V^J_M$ has the same eigenvalues. 
\end{proof}
The action of the Casimirs on the basis vectors of $V_J$
\begin{equation}
\ket{(j)\,J}:=\ket{(j)\,J,J},\quad j\in\Omega_J(\lambda,\gamma)
\end{equation}
is given by
\begin{subequations}
\label{eq:casimir_actions}
\begin{align}
\begin{split}
\cC_1\ket{(j)\,J}=&\left[J(J+1)\left(\half[\ii A + P_{\lambda,\rho}(j)]\right)-\Big(j(j+1)-\gamma(\gamma+1)\Big)\left(\half[\ii A - P_{\lambda,\rho}(j)]\right)\right]\ket{(j)\,J}\\
&+\half[P^+_{\lambda,\rho}(j)]\sqrt{J+j+\gamma+2}\sqrt{j+\gamma-J+1}\sqrt{J+j-\gamma+1}\sqrt{J-j+\gamma}\ket{(j+1)\,J}\\
&+\half[P^-_{\lambda,\rho}(j)]\sqrt{J+j+\gamma+1}\sqrt{j+\gamma-J}\sqrt{J+j-\gamma}\sqrt{J-j+\gamma+1}\ket{(j-1)\,J}
\end{split}
\\[0.8em]
\begin{split}
\cC_2\ket{(j)\,J}=&\left[\Big(J(J+1)-j(J+1)\Big)\Big(1 -\ii A P_{\lambda,\rho}(j)\Big)+\gamma(\gamma+1)\Big(1 +\ii A P_{\lambda,\rho}(j)\Big)+\lambda^2 +\rho^2 -1\right]\ket{(j)\,J}\\
&-\ii A P^+_{\lambda,\rho}(j)\sqrt{J+j+\gamma+2}\sqrt{j+\gamma-J+1}\sqrt{J+j-\gamma+1}\sqrt{J-j+\gamma}\ket{(j+1)\,J}\\
&-\ii A P^-_{\lambda,\rho}(j)\sqrt{J+j+\gamma+1}\sqrt{j+\gamma-J}\sqrt{J+j-\gamma}\sqrt{J-j+\gamma+1}\ket{(j-1)\,J},
\end{split}
\end{align}
\end{subequations}
where it is implicitly assumed that $\ket{(j)\,J}=\vec{0}$ if $j\not\in \Omega_J(\lambda,\gamma)$.
Note that the matrix form of each $\cC_a$ is \emph{tridiagonal} (see Appendix~\ref{app:tridiagonal}) and that, for the subdiagonal entries,
\begin{equation}
\braket{(j+1)\,J|\cC_a|(j)\,J}=0\quad\Leftrightarrow\quad j=J+\gamma=\max\Omega_J(\lambda,\gamma);
\end{equation}
it follows from Proposition~\ref{prop:tridiagonal} that the eigenspaces of $\cC_a$ are all $1$-dimensional, so that it is diagonalisable if and only if it has $\dim V_J$ \emph{distinct} eigenvalues.
Explicitly, the Casimirs are simultaneously diagonalisable on $V_J$ if and only if there is a basis
\begin{equation}
\ket{(\Lambda,\Rho)\,J}=\sum_{j\in\Omega_J}\mathrm{A}\lbrace\gamma_A;(\lambda,\rho)\,j|(\Lambda,\Rho)\,J\rbrace\ket{(j)\,J},\quad (\Lambda,\Rho)\in\cC_J(\lambda,\rho,\gamma,A)\subseteq \C^2,
\end{equation}
with
\begin{equation}
|\cC_J(\lambda,\rho,\gamma,A)|=\dim V_J,
\end{equation}
such that
\begin{subequations}
\label{eq:lambda_rho_eigenvectors}
\begin{align}
\cC_1\ket{(\Lambda,\Rho)\,J}&=i\Lambda\Rho\ket{(\Lambda,\Rho)\,J}\\
\cC_2\ket{(\Lambda,\Rho)\,J}&=(\Lambda^2+\Rho^2-1)\ket{(\Lambda,\Rho)\,J}
\end{align}
\end{subequations}
and for every $(\Lambda,\Rho)$, $(\Lambda',\Rho')\in \cC_J(\lambda,\rho,\gamma,A)$
\begin{equation}
\begin{cases}
\Lambda\Rho=\Lambda'\Rho'\\
\Lambda^2+\Rho^2=(\Lambda')^2+(\Rho')^2
\end{cases}
\quad \Leftrightarrow \quad
(\Lambda',\Rho')=(\Lambda,\Rho);
\end{equation}
note that at this stage $\Lambda$ is allowed to be any complex number, to ensure that any pair of eigenvalues of the Casimirs can be written as in~\eqref{eq:lambda_rho_eigenvectors}.
The coefficients of the change of basis $\mathrm{A}\lbrace\gamma_A;(\lambda,\rho)\,j|(\Lambda,\Rho)\,J\rbrace$ will be called \emph{\CG\ coefficients}, in analogy with $\su(2)$ representation theory.
Conversely, one has the inverse change of basis
\begin{equation}
\ket{(j)\,J}=\sum_{(\Lambda,\Rho)\in\cC_J}\mathrm{B}\lbrace (\Lambda,\Rho)\,J|\gamma_A;(\lambda,\rho)\,j \rbrace\ket{(\Lambda,\Rho)\,J},\quad j\in\Omega_J(\lambda,\rho),
\end{equation}
where the $\mathrm{B}\lbrace (\Lambda,\Rho)\,J|\gamma_A;(\lambda,\rho)\,j \rbrace$ will be called \emph{inverse \CG\ coefficients}.
As a consequence of Proposition~\ref{prop:JM}, the eigenvectors in $V^J_M$, $M<J$ will be
\begin{equation}
\ket{(\Lambda,\Rho)\,J,M}:=\sum_{j\in\Omega_J}\mathrm{A}\lbrace\gamma_A;(\lambda,\rho)\,j|(\Lambda,\Rho)\,J\rbrace\ket{(j)\,J,M}
\end{equation}
so that, more generally,
\begin{equation}
\ket{(\Lambda,\Rho)\,J,M}=\sum_{j\in\Omega_J}\sum_{\mu\in\cM_\gamma}\sum_{m\in\cM_j}
\mathrm{A}\lbrace\gamma_A;(\lambda,\rho)\,j|(\Lambda,\Rho)\,J\rbrace
\braket{\gamma,\mu;j,m|J,M}\ket{\gamma_A,\mu}\otimes\ket{(\lambda,\rho)\,j,m}.
\end{equation}

Solving the eigenvalue equations for arbitrary $\gamma$ is not an easy task: instead, we will solve explicitly the case $\gamma=\half$ and proceed by induction for the other cases.
When $\gamma=\half$, one has
\begin{equation}
\mathcal{J}(\lambda,\half)=
\begin{cases}
\half+\N_0&\casesif\lambda=0\\
|\lambda|-\half+\N_0&\casesif\lambda\neq0
\end{cases}
\quad\textnormal{and}\quad \dim V_J=
\begin{cases}
1&\casesif J=|\lambda|-\half\\
2&\casesif J\geq|\lambda|+\half,
\end{cases}
\end{equation}
and it can be explicitly checked that, when $\lambda\neq-A\rho$,
\begin{equation}
\cC_J(\lambda,\rho,\half,A)=
\begin{cases}
\big\lbrace (\lambda-\half,\rho-\half[A]),(\lambda+\half,\rho+\half[A]) \big\rbrace \subseteq\Z/2 \times \C &\casesif J\geq|\lambda|+\half\\
\big\lbrace (\lambda-\half\sgn(\lambda),\rho-\half[A]\sgn(\lambda)) \big\rbrace \subseteq\Z/2 \times \C  &\casesif J=|\lambda|-\half;
\end{cases}
\end{equation}
the corresponding \CG\ coefficients can be found in Table~\ref{tab:1/2} (Appendix~\ref{app:CG}). When $\rho=-A\lambda$ the eigenvalues for $J\geq |\lambda|+\half$ coincide, so that, as pointed out earlier, the Casimirs cannot be diagonalised.

As the eigenvalues do not depend on $J$ and
\begin{equation}
\cC_J(\lambda,\rho,\half,A)\subseteq\cC_{J+1}(\lambda,\rho,\half,A),\quad \forall J\in\cJ(\lambda,\gamma),
\end{equation}
the eigenvectors can be extended to an eigenbasis 
\begin{equation}
\ket{(\Lambda,\Rho)\,J,M},\quad (\Lambda,\Rho)=(\lambda\pm\half,\rho\pm\half[A]),\quad J\in|\Lambda|+\N_0,\quad M\in \cM_J
\end{equation}
of the whole product space, where~\eqref{eq:sigma_decomposition} ensures that the counting is consistent. One can check that, for all eigenvalue pairs,
\begin{equation}
\Rho\not\in\pm(|\Lambda|+\N),
\end{equation}
so that $F^A_{\shalf}\otimes V_{\lambda,\rho}$ splits in two infinite-dimensional irreducible modules
\begin{equation}
V_{\Lambda,\Rho},\quad (\Lambda,\Rho)=(\lambda\pm\half,\rho\pm\half[A]).
\end{equation}
These modules are \emph{never} both unitary: a list of the possible pairs $(\lambda,\rho)$ such that there  is \emph{one} unitary module in the decomposition can be found in Table~\ref{tab:LR_unitary}. Notice that there are, up to isomorphisms, only two unitary modules that coupled with $F^A_{\shalf}$ have a unitary one in the decomposition.
\begin{table}[!ht]
\centering
\begin{tabular}{lcccccc}
\toprule[\lightrulewidth]
\toprule[\lightrulewidth]
& $\qquad\qquad$ & $\lambda$ & $\quad$ & $\rho$ & $\quad$ & $V_{\lambda,\rho}$ unitary if  \\
\midrule
principal series && any && $\pm \half + \ii\R$ && $(\lambda,\rho)=(0,\pm\half)$\\
\addlinespace[0.8em]
complementary series && $\pm\half$ && $\sgn(\lambda)\half[A]+(-1,0)\cup(0,1)$ && $(\lambda,\rho)=(\pm\half,0)$ \\
\bottomrule[\lightrulewidth]
\bottomrule[\lightrulewidth]
\end{tabular}
\caption{The possible pairs $(\lambda,\rho)$ such that one $V_{\Lambda,\Rho}$ is unitary (principal or complementary series).}
\label{tab:LR_unitary}
\end{table}

A generalisation of the case $\gamma=\half$ to arbitrary $\gamma$ is provided by the following Propositions:
\begin{lemma}
\label{lem:no_double_eigenvalues}
Let $\gamma\in\N/2$, $A=\pm1$ and $(\lambda,\rho)\in\Z/2\times\C$, with $\rho\not\in\pm(|\lambda|+\N)$. Then
\begin{equation*}
\begin{cases}
\ii(\lambda+\mu)(\rho+A\mu)\neq \ii(\lambda+\nu)(\rho+A\nu)\\[0.5em]
(\lambda+\mu)^2+(\rho+A\mu)^2-1\neq (\lambda+\nu)^2+(\rho+A\nu)^2-1
\end{cases}
\quad \forall \mu\neq\nu\in\cM_\gamma
\end{equation*}
if and only if $\rho+A\lambda\not\in(-2\gamma,2\gamma)\cap\Z$.
\end{lemma}
\begin{proof}
Let $\mu\neq\nu\in\cM_\gamma$. The statement reduces to
\begin{equation}
\mu^2+\mu(\lambda+A\rho)\neq \nu^2+\nu(\lambda+A\rho),
\end{equation}
which is equivalent to
\begin{equation}
\label{eq:no_double_eigenvalues}
(\mu-\nu)(\lambda+A\rho+\mu+\nu)\neq 0\quad \Leftrightarrow \quad \lambda+A\rho+\mu+\nu\neq0.
\end{equation}
The possible values the sum $\mu+\nu$ can take are
\begin{equation}
\{\mu+\nu \setst \mu\neq\nu\in \cM_\gamma \}\equiv \left\{-2\gamma+1,-2\gamma+2,\dotsc,2\gamma-2,2\gamma-1 \right\},
\end{equation}
so that~\eqref{eq:no_double_eigenvalues} is true if and only if $\rho+A\lambda\not\in(-2\gamma,2\gamma)\cap \Z$.
\end{proof}
\begin{proposition}
\label{prop:coupling}
Consider the product $F^A_\gamma\otimes V_{\lambda,\rho}$, with $\gamma\geq\half$ and $V_{\lambda,\rho}$ infinite-dimensional. When $\rho+A\lambda\not\in(-2\gamma,2\gamma)\cap \Z$ the Casimirs are simultanously diagonalisable, with
\begin{equation*}
(\Lambda,\Rho)\in\lbrace (\lambda+\nu,\rho+A\nu)\setst \nu\in\cM_\gamma\rbrace.
\end{equation*}
\end{proposition}
\begin{proof}
The proof proceeds by induction on $\gamma\in\N/2$. Assume that the statement is true for $\gamma-\half$, and consider the product $F^A_\gamma\otimes V_{\lambda,\rho}$, $\gamma>\half$.

It is known from $\su(2)$ representation theory that
\begin{equation}
F^A_{\shalf}\otimes F^A_{\gamma-\shalf}=F^A_{\gamma-1}\oplus F^A_\gamma,
\end{equation}
so that
\begin{equation}
\ket{\gamma_A,\mu}\equiv \sum_{\sigma\in\cM_{\shalf}}\sum_{\tau\in\cM_{\gamma-\shalf}}\braket{\half,\sigma;\gamma-\half,\tau|\gamma,\mu}\ket{\half_A,\sigma}\otimes\ket{(\gamma-\half)_A,\tau};
\end{equation}
in particular
\begin{equation}
\label{eq:1/2+gamma-1/2}
\ket{\gamma_A,\gamma}=\ket{\half_A,\half}\otimes\ket{(\gamma-\half)_A,\gamma-\half}.
\end{equation}
Consider now the $J^2$-eigenspace $V_J$, $J\geq |\lambda|+\gamma$,
so that $J-\gamma\in\Omega_J(\lambda,\gamma)$ and the vector
\begin{equation}
\ket{(J-\gamma)\,J}=\ket{\gamma_A,\gamma}\otimes\ket{(\lambda,\rho)\, J-\gamma, J-\gamma}
\end{equation}
exists. Using~\eqref{eq:1/2+gamma-1/2}, $\ket{(J-\gamma)\,J}$ can be rewritten as
\begin{equation}
\begin{split}
\ket{(J-\gamma)\,J}&=\ket{\half_A,\half}\otimes\Big( \ket{(\gamma-\half)_A,\gamma-\half}\otimes\ket{(\lambda,\rho)\, J-\gamma, J-\gamma}\Big)\\
&=\sum_{\tau\in\cM_{\gamma-\shalf}}
\mathrm{B}\lbrace (\lambda+\tau,\rho+A\tau)\,J-\half | (\gamma-\half)_A;(\lambda,\rho)\,J-\gamma \rbrace
\ket{\half_A,\half}\otimes\ket{(\lambda+\tau,\rho+A\tau)\,J-\half},
\end{split}
\end{equation}
where the inductive hypothesis and the fact that
\begin{equation}
\rho+A\lambda\not\in(-2\gamma,2\gamma)\cap \Z\quad \Rightarrow \quad \rho+A\lambda\not\in(-2\gamma+1,2\gamma-1)\cap \Z
\end{equation}
were used.
Since in particular $\rho\neq-A\lambda$, the results of the case $\gamma=\half$ can be used, so that 
\begin{multline}
\label{eq:decomposition 1}
\ket{(J-\gamma)\,J}=\sum_{\sigma\in\cM_{\shalf}}
\sum_{\tau\in\cM_{\gamma-\shalf}}
\mathrm{B}\lbrace (\lambda+\tau+\sigma,\rho+A\tau+A\sigma)\,J | \half_A;(\lambda+\tau,\rho+A\tau)\,J-\half \rbrace\\
\times\mathrm{B}\lbrace (\lambda+\tau,\rho+A\tau)\,J-\half | (\gamma-\half)_A;(\lambda,\rho)\,J-\gamma \rbrace \ket{[\lambda+\tau](\lambda+\tau+\sigma,\rho+A\tau+A\sigma)\,J},
\end{multline}
where $[\sigma]$ keeps track of the fact that $(\lambda+\tau+\sigma,\rho+A\tau+A\sigma)$ comes from $(\lambda+\tau,\rho+A\tau)$. There are 
exactly $4\gamma$ (independent) vectors on the rhs of~\eqref{eq:decomposition 1}, namely
\begin{equation}
\begin{cases}
\ket{[+\half](\lambda+\gamma,\rho+A\gamma)\,J}\\
\ket{[-\half](\lambda+\nu,\rho+A\nu)\,J}\quad\textnormal{and}\quad\ket{[+\half](\lambda+\nu,\rho+A\nu)\,J},\quad \nu \in \cM_{\gamma-1}\\
\ket{[-\half](\lambda-\gamma,\rho-A\gamma)\,J},
\end{cases}
\end{equation}
with $2\gamma+1\equiv \dim V_J$ distinct eigenvalue pairs (see Lemma~\ref{lem:no_double_eigenvalues}).

As shown in Appendix~\ref{app:CG}, when $J\geq |\lambda|+\gamma$ the \CG\ coefficients satisfy
\begin{equation}
\frac{\mathrm{B}\lbrace(\Lambda+1,\Rho+A)\,J-\half|(\gamma-\half)_A;(\lambda,\rho)\,J-\gamma\rbrace}{\mathrm{B}\lbrace(\Lambda,\Rho)\,J-\half|(\gamma-\half)_A;(\lambda,\rho)\,J-\gamma\rbrace}=
\alpha(\Lambda,\Rho)\frac{\sqrt{J+\Lambda+\half}\sqrt{J+A\Rho+\half}}{\sqrt{J-\Lambda-\half}\sqrt{J-A\Rho-\half}},
\end{equation}
where $\alpha$ is fixed by the normalisation convention and is independent of $J$. Using this formula and the $\gamma=\half$ \CG\ coefficients from Table~\ref{tab:1/2}
\begin{equation}
\mathrm{B}\lbrace (\Lambda+\sigma,\Rho+A\sigma)\,J | \half_A;(\Lambda,\Rho)\,J-\half \rbrace=
\begin{cases}
\ii A\frac{\sqrt{J-\Lambda+\shalf}\sqrt{J-A\Rho+\shalf}}{\sqrt{2J+1}\sqrt{\Lambda+A\Rho}} & \casesif \sigma=-\half
\\
\frac{\sqrt{J+\Lambda+\shalf}\sqrt{J+A\Rho+\shalf}}{\sqrt{2J+1}\sqrt{\Lambda+A\Rho}} & \casesif \sigma=+\half,
\end{cases}
\end{equation}
it is possible to write
\begin{equation}
\ket{(J-\gamma)\,J}=\sum_{\nu\in\cM_\gamma}\mathrm{B}\lbrace (\lambda+\nu,\rho+A\nu)\,J | \gamma_A; (\lambda,\rho)\,J-\gamma \rbrace \ket{(\lambda+\nu,\rho+A\nu)\,J},
\end{equation}
where the vectors on the rhs are defined (up to a normalisation factor) as
\begin{equation}
\ket{(\Lambda,\Rho)\,J}:=
\begin{cases}
\ket{[+\half](\Lambda,\Rho)\,J} & \casesif (\Lambda,\Rho)=(\lambda+\gamma,\rho+A\gamma)\\
\ket{[-\half](\Lambda,\Rho)\,J} & \casesif (\Lambda,\Rho)=(\lambda-\gamma,\rho-A\gamma)\\
\frac{1}{\sqrt{\Lambda+A\Rho-1}}\ket{[+\half](\Lambda,\Rho)\,J}+\ii A \frac{\alpha(\Lambda-\shalf,\Rho-\shalf[A])}{\sqrt{\Lambda+A\Rho+1}} \ket{[-\half](\Lambda,\Rho)\,J}& \quad \textnormal{otherwise}.
\end{cases}
\end{equation}
As these vectors live in different $(\cC_1,\cC_2)$-eigenspaces, they are necessarily independent. Moreover, they form a basis of $V_J$: in fact, we know from~\eqref{eq:casimir_actions} that
\begin{equation}
\cC_1\ket{(j)\, J}\in\Span\lbrace \ket{(j-1)\, J},\ket{(j)\, J},\ket{(j+1)\, J} \rbrace,
\end{equation}
with
\begin{equation}
\braket{(j-1)\, J|\cC_1|(j)\, J}=0\quad\Leftrightarrow\quad j=J+\gamma,
\end{equation}
so that
\begin{equation}
\ket{(j+1)\, J}\in\Span\lbrace \ket{(j-1)\, J},\ket{(j)\, J},\cC_1\ket{(j)\, J} \rbrace.
\end{equation}
Since $\ket{(J-\gamma)\, J}$ is a linear combination of the $\ket{(\Lambda,\Rho)\, J}$ vectors, it follows recursively that
\begin{equation}
\ket{(j)\, J}\in\Span\lbrace \ket{(\lambda+\nu,\rho+A\nu)\, J}\setst 
\nu\in\cM_\gamma \rbrace,\quad \forall j\in\Omega_J(\lambda,\gamma).
\end{equation}

It follows from Proposition~\ref{prop:JM} that all the results obtained for $V_J$ hold for each $V^J_M$, $M\in\cM_J$. One can then extend them to every $J\in \cJ(\lambda,\gamma)$ by defining recursively (once the \CG\ coefficients  and the vectors have been appropriately normalised, so that the generators act on the $(\cC_1,\cC_2)$-eigenvectors as~\eqref{eq:K_actions})
\begin{equation}
\ket{(\Lambda,\Rho)\,J-1,J-1}\propto K_-\ket{(\Lambda,\Rho)\,J,J}-P_{\Lambda,\Rho}(J)\sqrt{2J}\ket{(\Lambda,\Rho)\,J,J-1}-P^+_{\Lambda,\Rho}(J)\sqrt{2}\ket{(\Lambda,\Rho)\,J+1,J-1},
\end{equation}
for $J\leq|\lambda|-\gamma$, which are trivially still eigenvectors. Since a basis of eigenvectors has been constructed for the whole product space, it follows that the Casimirs are diagonalisable.
\end{proof}

The results of Proposition~\ref{prop:coupling} does not apply when $\lambda=-A\rho$; in this case one has
\begin{proposition}
\label{prop:continuity}
Consider the product $F^A_\gamma\otimes V_{\lambda,\rho}$, with $\gamma\geq\half$ and $V_{\lambda,\rho}$ infinite-dimensional.
When $\rho+A\lambda\not\in(-2\gamma,2\gamma)\cap \Z$ the Casimirs are not diagonalisable on the product module.
\end{proposition}
\begin{proof}
Consider the $J^2$-eigenspace $V_J$, $J\geq|\lambda|+\gamma$, and let
\begin{equation}
d(\lambda,\rho,k):=\det(\cC_1|_{V_J}-k\1),\quad k\in\C,
\end{equation}
which is continuous in $\rho\in \C$, as it is a product of continuous function of $\rho$ (see eq.~\eqref{eq:casimir_actions}). From Proposition~\ref{prop:coupling} we have that, for $\rho+A\lambda\not\in(-2\gamma,2\gamma)\cap \Z$,
\begin{equation}
d(\lambda,\rho,k)=\prod_{\nu\in\cM_\gamma}\big[\ii(\lambda+\nu)(\rho+A\nu)-k\big];
\end{equation}
it follows from continuity that, for each fixed $\lambda\in \Z/2$, $n\in \{-2\gamma+1,\dotsc,2\gamma-1\}$,
\begin{equation}
d(\lambda,-A\lambda+n,k)=\lim_{\rho\rightarrow -A\lambda+n}d(\lambda,\rho,k).
\end{equation}
From Lemma~\ref{lem:no_double_eigenvalues} we know that there are  \emph{at most} $2\gamma$ distinct eigenvalues in this case, while $\dim V_J=2\gamma+1$. As pointed out earlier, the matrix form of $\cC_a|_{V_J}$ satisfies the assumptions of Proposition~\ref{prop:tridiagonal}, so that it has at most $2\gamma$ eigenvectors, i.e., it is not diagonalisable on $V_J$ (and hence on the whole product space). 
\end{proof}

To summarize, a \CG\ decomposition of the product $F^A_\gamma\otimes V_{\lambda,\rho}$ is possible if and only if $\rho\neq-A\lambda$, with the modules in the decomposition having
\begin{equation}
(\Lambda,\Rho)\in\lbrace(\lambda+\nu,\rho+A\nu) \setst \nu\in\cM_\gamma\rbrace.
\end{equation}
Finally, the result for left and right modules can be generalised to arbitrary finite-dimensional ones in the following
\begin{corollary}
The Casimirs are simultaneously diagonalisable in $(\gamma_1,\gamma_2)\otimes V_{\lambda,\rho}$, with $\gamma_1,\gamma_2\geq \shalf$ and $V_{\lambda,\rho}$ infinite-dimensional, if and only if $\rho-\lambda\not\in(-2\gamma_1,2\gamma_1)\cap\Z$ and $\rho+\lambda\not\in(-2\gamma_2,2\gamma_2)\cap\Z$, with
\begin{equation*}
(\Lambda,\Rho)\in\lbrace(\lambda+\nu_1+\nu_2,\rho-\nu_1+\nu_2) \setst \nu_1\in \cM_{\gamma_1}, \nu_2\in \cM_{\gamma_2}\rbrace;
\end{equation*}
the eigenvalue pairs are not necessarily distinct.
\end{corollary}
\begin{proof}
As already noted, one has $(\gamma_1,\gamma_2)\cong F^{-}_{\gamma_1}\otimes F^{+}_{\gamma_2}$, so that we may diagonalise the Casimirs in $F^{+}_{\gamma_2}\otimes V_{\lambda,\rho}$ first, and then, for each resulting eigenspace $V_{\lambda,\rho}$, in $F^{-}_{\gamma_1}\otimes V_{\Lambda,\Rho}$. We can distinguish $3$ cases:
\begin{enumerate}
\item if $\rho+\lambda\not\in(-2\gamma_2,2\gamma_2)\cap\Z$ the product $F^{+}_{\gamma_2}\otimes V_{\Lambda,\Rho}$ admits a decomposition. The second decomposition exists if and only if, for each $\nu\in\cM_{\gamma_2}$,
\begin{equation}
\rho+\nu-(\lambda+\nu)=\rho-\lambda\not\in(-2\gamma_1,2\gamma_1)\cap\Z.
\end{equation}
\item If $\rho+\lambda\in(-2\gamma_2,2\gamma_2)\cap\Z$ but $\rho-\lambda\not\in(-2\gamma_1,2\gamma_1)\cap\Z$  the product $F^{+}_{\gamma_2}\otimes V_{\lambda,\rho}$ is not decomposable, but one can  use the fact that $F^{-}_{\gamma_1}\otimes F^{+}_{\gamma_2}\cong F^{+}_{\gamma_2}\otimes F^{-}_{\gamma_1}$ and decompose the product $F^{-}_{\gamma_1}\otimes V_{\Lambda,\Rho}$ first. Following the same reasoning of the previous case, the product of each resulting submodule with $F^+_{\gamma_2}$ will not be decomposable, as for each $\nu\in\cM_{\gamma_1}$
\begin{equation}
\rho+\nu+(\lambda-\nu)=\rho+\lambda\in(-2\gamma_2,2\gamma_2)\cap\Z.
\end{equation}
\item Finally, if $\rho-\lambda\in(-2\gamma_1,2\gamma_1)\cap\Z$ and$\rho+\lambda\in(-2\gamma_2,2\gamma_2)\cap\Z$, both $F^{-}_{\gamma_1}\otimes V_{\Lambda,\Rho}$ and  $F^{+}_{\gamma_2}\otimes V_{\Lambda,\Rho}$ are non-decomposable. The only results we have are on the product of $F^A_\gamma$ with \emph{irreducible} modules, so we are not in a position to say anything in this case. However, one can check, following the same procedures of Proposition~\ref{prop:coupling}, that $F^+_{\gamma_2}\otimes V_{\lambda,\rho}$ has at least one submodule $V_{\lambda+\gamma_2,\rho+\gamma_2}$. Since the product $F^-_{\gamma_1}\otimes V_{\lambda+\gamma_2,\rho+\gamma_2}$ is not decomposable, the product $(\gamma_1,\gamma_2)\otimes V_{\lambda,\rho}\supseteq F^-_{\gamma_1}\otimes V_{\lambda+\gamma_2,\rho+\gamma_2}$ will be indecomposable as well.
\end{enumerate}
 The values of the Casimirs follow from Proposition~\ref{prop:coupling}, and it can be checked explicitly that they need not be all distinct: for example, when $\gamma_1=\gamma_2=1$, two possible pairs are $(\lambda,\rho)$ and $(\lambda-2,\rho)$, which are equivalent if $(\lambda,\rho)=(1,0)$.
\end{proof}
The \CG\ coefficients of the product $(\gamma_1,\gamma_2)\otimes V_{\lambda,\rho}$ will be denoted by
\begin{equation}
\mathrm{A}\{ (\gamma_1,\gamma_2)\,\mu_1,\mu_2;(\lambda,\rho)\,j,m | [\alpha] (\Lambda,\Rho)\, J,M\},
\end{equation}
where $\alpha$ keeps track of the multiplicity of $V_{\Lambda,\Rho}$ in the decomposition.
\section{Tensor Operators}\label{sec:tensor_operators}
This sections contains the main results of the paper, i.e., the generalisation of the \WE\ theorem to arbitrary groups and in particular to the infinite-dimensional representations of the Lorentz group.
The section is organized as follows: first tensor operators for group representations and $(\mathfrak{g},K)$-modules of arbitrary Lie groups are introduced, then \WE\ theorem will be presented. Lastly, as an application, the theorem will be used to generalise the \JS\ representation to infinite-dimensional $(\mathfrak{g},K)$-modules of the Lorentz group.
\subsection{Definition of tensor operators}\label{subsec:tensor_operators}

Tensor operators for a Lie group $G$ are, roughly speaking, a particular class of operators that transform as a vector in a representation of $G$ under the action of the group. Explicitly
\begin{definition}[Strong tensor operator]
Let $V_0$, $V$ and $V'$ be arbitrary (topological) $G$-modules of a Lie group $G$, with $V_0$ finite dimensional. A \emph{strong} tensor operator for $G$ is an intertwiner between $V_0\otimes V$ and $V'$, i.e. a \emph{continuous} linear map
\[
T:V_0\otimes V \rightarrow V'
\]
such that
\[
T\circ g = g \circ T,\quad \forall g\in G.
\]
If $V_0$ is irreducible $T$ is called an irreducible strong tensor operator.
\end{definition}
A weaker definition can be used when working with $(\mathfrak{g},K)$-modules instead of group representations: 
\begin{definition}[tensor operator]
Let $V_0$, $V$ and $V'$ be arbitrary $(\mathfrak{g},K)$-modules of a Lie group $G$, with $V_0$ finite dimensional. A tensor operator for $G$ is an intertwiner between $V_0\otimes V$ and $V'$, i.e., a linear map
\[
T:V_0\otimes V \rightarrow V'
\]
such that
\[
T \circ X= X\circ T,\quad \forall X\in \mathfrak{g}
\quad \textnormal{and} \quad
T \circ k = k \circ T ,\quad \forall k \in K,
\]
where $\mathfrak{g}$ and $K$ act on the product module as
\begin{align*}
X(v_0\otimes v)&=(X\,v_0)\otimes  v + v_0 \otimes (X\, v)\\
k(v_0\otimes v)&=(k\,v_0)\otimes (k\, v).
\end{align*}
If $V_0$ is irreducible $T$ is called an irreducible tensor operator.
\end{definition}
The two definitions of tensor operator are related to each other; one has, in  fact
\begin{proposition}
An intertwiner $T:V\rightarrow V'$ between $G$-modules is also an intertwiner between the corresponding $(\mathfrak{g},K)$-modules. As a consequence, a strong tensor operator is also a tensor operator.
\end{proposition}
\begin{proof}
The subspace $V_K\subseteq V$ of $K$-finite vectors, i.e., the set of all vectors $v$ such that $\Span\{k\,v,k\in K\}$  is finite-dimensional, is the $(\mathfrak{g},K)$-module associated to $V$, with
\begin{equation}
X\,v:=\left.\frac{d}{dt}\right|_{t=0}\exp(tX)v,\quad \forall X\in \mathfrak{g},\quad \forall v\in V_K.
\end{equation}
One has, since $T$ commutes with the action of $K\subseteq G$,
\begin{equation}
\dim\Span\{k T v \setst k\in K \}=\dim\Span\{ Tk v \setst k\in K \}=\dim T\left(\Span\{k v \setst k\in K \}\right)<\infty,\quad \forall v\in V_K,
\end{equation}
that is $T(V_K)\subseteq V'_K$. Moreover,
\begin{equation}
XTv=\left.\frac{d}{dt}\right|_{t=0}\exp(tX)Tv=\left.\frac{d}{dt}\right|_{t=0}T\exp(tX)v=T\left.\frac{d}{dt}\right|_{t=0}\exp(tX)v=TXv,\quad \forall X\in \mathfrak{g},\quad \forall v\in V_K,
\end{equation}
where the fact that $T$ is continuous was used. It follows that $T|_{V_K}$ is an intertwiner between the $(\mathfrak{g},K)$-modules $V_K$ and $V'_K$.
\end{proof}

It is often preferable to have operators between $V$ and $V'$: this can be achieved by defining the ``components'' of a (strong) tensor operator $T$ in a basis $e_i\in V_0$, $i\in I$ as
\begin{equation}
T_i:v\in V \mapsto T (e_i,v)\in V';
\end{equation}
the definitions of strong tensor operator and tensor operator become respectively
\begin{equation}
\label{eq:intertwiner_def_basis1}
gT_i g^{-1}=\sum_{j\in I}\braket{e^j,g\,e_i}T_j,\quad \forall g \in G,
\end{equation}
and
\begin{equation}
\label{eq:intertwiner_def_basis2}
[X,T_i]=\sum_{j\in I}\braket{e^j,X\,e_i}T_j,\quad \forall X \in \mathfrak{g}
\quad\textnormal{and}\quad
k\,T_i\, k^{-1}=\sum_{j\in I}\braket{e^j,k\,e_i}T_j,\quad \forall k \in K,
\end{equation}
where $\braket{\cdot,\cdot}$ is the dual pairing of $V_0^*$ and $V_0$ and $e^j\in V_0^*$, $j\in I$ is the dual basis defined by
\begin{equation}
\braket{e^j,e_i}=e^j(e_i)=\tensor{\delta}{^j_i}.
\end{equation}

The definition of tensor operator can be simplified when $K$ is connected, since one can simply require the the operator commutes with every element of $\mathfrak{g}$. In fact one has\footnote{The proof is based on~\cite{stackexchange}.}
\begin{proposition}
If $K$ is connected a linear map $T:V\rightarrow V'$ is a $(\mathfrak{g},K)$-module homomorphism if
\begin{equation*}
T \circ X= X\circ T,\quad \forall X\in \mathfrak{g}.
\end{equation*}
\begin{proof}
Let $\mathfrak{k}\subseteq\mathfrak{g}$ be the Lie algebra of $K$. For any $X\in \mathfrak{k}$, $v\in V$, $u\in V^*$ one has
\begin{equation}
\left.\frac{d}{dt}\right|_{t=0}\braket{u,\left( T\circ \exp(tX) - \exp(tX)\circ T \right)v}=\braket{u,\left( T\circ X - X\circ T \right)v}=0.
\end{equation}
Since the derivative vanishes, it must be
\begin{equation}
\braket{u,\left( T\circ \exp(X) - \exp(X)\circ T \right)v}=\braket{u,\left( T\circ \exp(0) - \exp(0)\circ T \right)v},\quad \forall v\in V,\quad\forall u \in V^*,
\end{equation}
so that 
\begin{equation}
T\circ \exp(X)=\exp(X)\circ T,\quad \forall X\in \mathfrak{g}.
\end{equation}
However, if $K$ is connected, $\exp(\mathfrak{k})\subseteq K$ generates $K$~\cite{kosmann2009}, hence
\begin{equation}
T\circ k = k \circ T,\quad \forall k \in K.
\end{equation} 
\end{proof}
\end{proposition}

\subsection{\WE\ theorem}

One of the most useful property of irreducible tensor operators is the so-called Wigner--Eckart theorem, originally proved for compact groups~\cite{barut} and later extended to non-compact groups~\cite{locallycompact} for the particular case of tensor operators transforming as (infinite-dimensional) \emph{unitary} representations, which we don't consider.

We will present here a version of the theorem working for a generic Lie group $G$. In order to prove it, the following Lemma is needed:
\begin{lemma}
Let $V_0$, $V$, $V$ be irreducible $(\mathfrak{g},K)$-modules for a Lie group $G$, with $V_0$ finite-dimensional. If a \CG\ decomposition into irreducible modules for $V_0\otimes V$ exists, a non-zero intertwiner
\begin{equation*}
T:V_0\otimes V\rightarrow V'
\end{equation*}
is possible if and only if $V'$ appears (at least once) in the decomposition.
If $\mathcal{T}$ is the vector space of all such intertwiners, $\dim \mathcal{T}$ equals the multiplicity of $V'$ in the decomposition, and a basis is provided by the projections in each of the submodules $W_\alpha\subseteq V_0\otimes V$, $W_\alpha\cong V'$, with $\alpha$ keeping track of the multiplicities.
\end{lemma}
\begin{proof}
Let $T:V_0\otimes V\rightarrow V'$ be a $(\mathfrak{g},K)$-module homomorphism.
Shur's Lemma for irreducible $(\mathfrak{g},K)$-modules~\cite{wallach} guarantees that, if $W\subseteq V_0\otimes V$ is a submodule,
\begin{equation}
T|_{W}\propto
\begin{cases}
\1 & \casesif W\cong V'\\
0 & \casestextn{otherwise}.
\end{cases}
\end{equation}
It is then trivial to see that any such $T$ can be written as a linear combinations of the independent maps $T^\alpha$ that project $V_0\otimes V$ on each $W_\alpha \cong V'$.
\end{proof}
We can now state the theorem, which trivially follow from the Lemma.
\begin{theorem*}[\WE]
Let $T:V_0\otimes V \rightarrow V'$ be an irreducible tensor operator, with $V$, $V'$ irreducible. If a \CG\ decomposition for $V_0\otimes V$ exists, $T$ is a linear combination of the projections $T^\alpha:V_0\otimes V\rightarrow W_\alpha$ into each irreducible component $W_\alpha\cong V'$. If $V'\not\subseteq V_0\otimes V$ the tensor operator must necessarily vanish.
\end{theorem*}
Note that, although the theorem works for a generic group, no information is available when $V_0\otimes V$ is not decomposable; moreover, knowledge of the decomposition is needed in order to gain any real information about the operator: it is for this reason that for non-compact groups, where the non-trivial decomposition of the product of a finite-dimensional and an infinite-dimensional is needed, one has to study it case by case.

We obtain in particular, for $\Spin(3,1)$, the following Corollaries, where the conditions for the decomposition to exist were studied in Section~\ref{sec:CG_decomposition}.
\begin{corollary}
Let $T:(\gamma_1,\gamma_2)\otimes V_{\lambda,\rho}\rightarrow V_{\lambda',\rho'}$ be a tensor operator for $\Spin(3,1)$, with $V_{\lambda,\rho}$, $V_{\lambda',\rho'}$ infinite-dimensional. If a \CG\ decomposition for $(\gamma_1,\gamma_2)\otimes V_{\lambda,\rho}$ exists, the matrix elements of T satisfy
\begin{equation*}
\braket{(\lambda',\rho')\,j',m'|T|(\gamma_1,\gamma_2)\,\mu_1,\mu_2;(\lambda,\rho)\,j,m}
=
\sum_\alpha
N_\alpha \mathrm{B}\{[\alpha] (\lambda',\rho')\, j',m'|(\gamma_1,\gamma_2)\,\mu_1,\mu_2;(\lambda,\rho)\,j,m\},
\end{equation*}
where $\alpha$ keeps track on the multiplicity of $V_{\lambda',\rho'}$ in the decomposition and the proportionality factors $N_\alpha$ do not depend on the particular vectors being evaluated. In particular, if $V_{\lambda',\rho'}$ does not appear in the \CG\ decomposition, T must identically vanish.
\end{corollary}

\begin{corollary}
\label{prop:WE_coroll}
Let $T:F^A_\gamma\otimes V_{\lambda,\rho}\rightarrow V_{\lambda',\rho'}$ be a tensor operator for $\Spin(3,1)$, with $V_{\lambda,\rho}$, $V_{\lambda',\rho'}$ infinite-dimensional. If a \CG\ decomposition for $F^A_\gamma\otimes V_{\lambda,\rho}$ exists, the matrix elements of T satisfy
\begin{equation*}
\braket{(\lambda',\rho')\,j',m'|T_\mu|(\lambda,\rho)\,j,m}\equiv
\braket{(\lambda',\rho')\,j',m'|T|\gamma_A,\mu;(\lambda,\rho)\,j,m}
\propto
\mathrm{B}\{ (\lambda',\rho')\, j'|\gamma_A;(\lambda,\rho)\,j\}\braket{j',m'|\gamma,\mu;j,m},
\end{equation*}
where the proportionality factor does not depend on the particular vectors being evaluated.
In particular, if $V_{\lambda',\rho'}$ does not appear in the \CG\ decomposition, T must identically vanish.
\end{corollary}

\subsection{\JS\ representation}

As an application of the \WE\ theorem, we will present here a generalisation to to infinite-dimensional $\Spin(3,1)$ representations of a well-known result of $\SU(2)$ representation theory, known as \emph{\JS\ representation}~\cite{schwinger1952}:
\begin{proposition}[\JS\ representation for $\SU(2)$]
\label{prop:JS_SU(2)}
Let $F_\gamma$, $\gamma\in \N/2$ be the complex finite-dimensional irreducible $\su(2)$-module on which the Casimir $J^2$ has eigenvalue $\gamma(\gamma+1)$.
The generators of $\su(2)_\C$, when acting on $F_\gamma$, can be written as
\[
J_0=\half(a_1^\dagger a_1 - a_2^\dagger a_2),\quad J_+=a_1^\dagger a_2,\quad J_-=a_2^\dagger a_1,
\]
where
\[
a_i\in\Lin(F_\gamma,F_{\gamma-\shalf}),\quad a_i^\dagger\in\Lin(F_\gamma,F_{\gamma+\shalf})
\]
and
\[
[a_i,a_j^\dagger]=\delta_{ij},\quad [a_i,a_j]=[a^\dagger_i,a^\dagger_j]=0.
\]
The operators $a_i$, $a_i^\dagger$, $i=1,2$ and $\1$ act on the algebraic direct sum $\bigoplus_{\gamma\in \N_0/2}F_\gamma$ as the (complexified) $5$-dimensional Heisenberg algebra $\mathfrak{h}_2(\R)_\C$, making it a unitary irreducible $\mathfrak{h}_2(\R)$-module.
\end{proposition}
A proof can be found in~\cite{wigner_eckart}.
The $a_i$, $a_i^\dagger$ operators are extensively studied in quantum mechanics, where they are known as \emph{quantum harmonic oscillators}~\cite{messiah2}.

The extension of this result to finite-dimensional $\Spin(3,1)$ representations trivially follows from the fact that $\spin(3,1)_\C\cong \su(2)_\C \oplus \su(2)_\C$ (see Section~\ref{sec:representation_theory}). A generalisation to infinite-dimensional $(\mathfrak{g},K)$-modules can be obtained by making use of tensor operators as follows.
\begin{proposition}
Let $\vec{M}^A=\half(\vec{J}-\ii A \vec{K})$, $A=\pm1$ be the generators of $\su(2)_\C\oplus\su(2)_\C\cong\spin(3,1)_\C$. There exist four tensor operators
\[
T^A:F^A_{\shalf}\otimes V_{\lambda,\rho}\rightarrow V_{\lambda-\shalf,\rho-\shalf[A]},
\quad
\widetilde{T}^A:F^A_{\shalf}\otimes V_{\lambda,\rho}\rightarrow V_{\lambda+\shalf,\rho+\shalf[A]},
\quad A=\pm 1,
\]
where $V_{\lambda,\rho}$ is an arbitrary infinite-dimensional $(\mathfrak{g},K)$-module, such that
\[
M_0^A=-\half \big(T^A_-\widetilde{T}^A_+ + T^A_+\widetilde{T}^A_-\big),\quad M_\pm^A=\pm T^A_\pm \widetilde{T}^A_\pm
\]
when acting on an infinite-dimensional module with $\rho^2\neq\lambda^2$, with
\[
T^A_\pm\ket{(\lambda,\rho)\,j,m}:=T^A\ket{\half_A,\pm\half}\otimes\ket{(\lambda,\rho)\,j,m}.
\]
Their matrix elements are
\begin{subequations}
\label{eq:JS_matrix}
\begin{align}
\braket{(\lambda-\half,\rho-\half[A])\,j-\tfrac{1}{2},m\pm\tfrac{1}{2}|T^A_\pm|(\lambda,\rho)\,j,m}&=\pm\frac{\sqrt{j\mp m}\sqrt{j+\lambda}\sqrt{j+A\rho}}{\sqrt{2j}\sqrt{2j+1}}
\\[0.5em]
\braket{(\lambda-\half,\rho-\half[A])\,j+\tfrac{1}{2},m\pm\tfrac{1}{2}|T^A_\pm|(\lambda,\rho)\,j,m}&=\ii A\frac{\sqrt{j\pm m+1}\sqrt{j-\lambda+1}\sqrt{j-A\rho+1}}{\sqrt{2j+1}\sqrt{2j+2}}
\\[0.5em]
\braket{(\lambda+\half,\rho+\half[A])\,j-\tfrac{1}{2},m\pm\tfrac{1}{2}|\widetilde{T}^A_\pm|(\lambda,\rho)\,j,m}&=\mp\ii A\frac{\sqrt{j\mp m}\sqrt{j-\lambda}\sqrt{j-A\rho}}{\sqrt{2j}\sqrt{2j+1}}
\\[0.5em]
\braket{(\lambda+\half,\rho+\half[A])\,j+\tfrac{1}{2},m\pm\tfrac{1}{2}|\widetilde{T}^A_\pm|(\lambda,\rho)\,j,m}&=\frac{\sqrt{j\pm m+1}\sqrt{j+\lambda+1}\sqrt{j+A\rho+1}}{\sqrt{2j+1}\sqrt{2j+2}},
\end{align}
\end{subequations}
and they satisfy the commutation relations
\begin{equation}
\label{eq:JS_commutation}
[T^A_+,\widetilde{T}^B_-]=[\widetilde{T}^A_+,T^B_-]=\delta^{AB},\quad [T^A_\mu,T^B_\nu]=[\widetilde{T}^A_\mu,\widetilde{T}^B_\nu]=0.
\end{equation}
\end{proposition}
\begin{proof}
Consider the tensor operators $T^A$, $\widetilde{T}^A$ described above. As a consequence of Corollary~\ref{prop:WE_coroll}, it must be
\begin{subequations}
\begin{align}
\braket{(\lambda-\half,\rho-\half[A])\,J,M|T^A_\mu|(\lambda,\rho)\,j,m}&=t^A(\lambda,\rho)\,\mathrm{B}\{(\lambda-\half,\rho-\half[A])\,J|\half_A;(\lambda,\rho)\, j \}\braket{J,M|\half,\mu;j,m}
\\[0.5em]
\braket{(\lambda+\half,\rho+\half[A])\,J,M|\widetilde{T}^A_\mu|(\lambda,\rho)\,j,m}&=\widetilde{t}^A(\lambda,\rho)\,\mathrm{B}\{(\lambda+\half,\rho+\half[A])\,J|\half_A;(\lambda,\rho)\, j \}\braket{J,M|\half,\mu;j,m},
\end{align}
\end{subequations}
with $t^A$, $\widetilde{t}^A$ arbitrary functions of $\lambda$ and $\rho$.
Let now
\begin{equation}
V^A_0:=-\sqrt{2}M^A_0,\quad V^A_{\pm 1}:=\pm M^A_\pm;
\end{equation}
one can check that they are the components in the basis $\ket{1_A,\mu}$ of a tensor operator $V^A:F^A_1\otimes V_{\lambda,\rho}\rightarrow V_{\lambda,\rho}$.

Suppose, as an ansatz, that
\begin{equation}
\label{eq:ansatz}
V_\mu^A=\sum_{\mu_1\in\cM_{\shalf}}\sum_{\mu_2\in\cM_{\shalf}}\braket{\half,\mu_1;\half,\mu_2|1,\mu}T^A_{\mu_1}\widetilde T^A_{\mu_2};
\end{equation}
it is a standard result for $\SU(2)$ tensor operators~\cite{barut} (remember that $F^A_\gamma$ is also a $\SU(2)$ representation) that the rhs is indeed the $\mu$ component of a tensor operator transforming like $F^A_1$, so that the ansatz is consistent. Evaluating the \CG\ coefficients, one can rewrite~\eqref{eq:ansatz} as
\begin{equation}
M_0^A=-\half \big(T^A_-\widetilde{T}^A_+ + T^A_+\widetilde{T}^A_-\big),\quad M_\pm^A=\pm T^A_\pm \widetilde{T}^A_\pm.
\end{equation}

Comparing the possible matrix elements of both sides of~\eqref{eq:ansatz} one can explicitly check that they agree, i.e., the ansatz is verified, if and only if
\begin{equation}
t^A(\lambda+\half,\rho+\half[A])\,\widetilde{t}^A(\lambda,\rho)=\lambda+A\rho\neq 0.
\end{equation}
We choose here the particular solution
\begin{equation}
t^A(\lambda,\rho)=\widetilde{t}^A(\lambda,\rho)=\sqrt{\lambda+A\rho};
\end{equation}
with this choice we recover the matrix elements~\eqref{eq:JS_matrix} and, after some  tedious but simple calculations, the commutation relation~\eqref{eq:JS_commutation}.
\end{proof}

Note that the commutation relations~\eqref{eq:JS_commutation} are those of the Lie algebra $\mathfrak{h}_2(\R)_\C\oplus \mathfrak{h}_2(\R)_\C$, the same as the finite-dimensional case (it obviously follows from Proposition~\ref{prop:JS_SU(2)}); in the infinite-dimensional case, however, since $\vec{M}^A$ does \emph{not} act on $V_{\lambda,\rho}$ as a unitary $\su(2)$ representation, the $T^A$, $\widetilde{T}^A$ (with $A$ fixed) operators will not act unitarily as a Heisenberg algebra either.

\section{Concluding remarks}
Although the \WE\ theorem was generalised to arbitrary Lie groups, the actual information it provides still relies on the knowledge of the appropriate \CG\ decomposition.
The author hopes that the techniques developed in this paper for the Lorentz group can be used as a guide in the study of other groups. A possible topic for future research is the generalisation of the theorem to \emph{quantum groups}, in particular to deformed enveloping algebras $\mathcal{U}_q(\mathfrak{sl}(2,\R))$ and $\mathcal{U}_q(\spin(3,1))$; in this context, the tensor operators defined here admit a simple generalization to Hopf algebras by requiring them to intertwine Hopf algebra representations~\cite{rittenberg}.

One topic that remains to be investigated is the nature of the Heisenberg algebra representations provided by the tensor operators considered for the \JS\ representation; the author leaves it for further analysis.

The \WE\ theorem and the \JS\ representation for the special case of $\SL(2,\R)$ have found an application in quantum gravity~\cite{girelli_sellaroli}; it is hoped the generalisations provided here will have applications to physics as well.
\begin{acknowledgments}
The author would like to thank Florian Girelli for introducing him to the topic and providing useful insights.
\end{acknowledgments}
\appendix
\section{Notation}
\begin{center}
\begin{tabular*}{\textwidth}{l @{\extracolsep{\fill}} ll}
\toprule[\lightrulewidth]\toprule[\lightrulewidth]
\multicolumn{3}{l}{Braket notation}\\ \midrule
$\ket{\psi}$ & &	Vector in an inner product space\\
$\braket{\varphi|\psi}$ & & Inner product of $\ket\varphi$ and $\ket\psi$, antilinear in $\ket\varphi$\\
$A^\dagger$ && Hermitian adjoint of an operator $A$\\[1em]
\multicolumn{3}{l}{Representation theory}\\ \midrule
$\mathrm{Lin}(V_1,V_2)$	& &	vector space of linear maps $V_1\rightarrow V_2$\\
$V\otimes W$	& &	(orthogonal) tensor product of two modules $V$, $W$\\
$V\oplus W$	& &	(orthogonal) direct sum of two modules $V$, $W$\\[1em]
\multicolumn{3}{l}{Sets}\\ \midrule
$x+\mathbb{Y}$ & & set defined by $\left\{ x+y \setst y\in\mathbb{Y} \right\}$\\
$x\mathbb{Y}$ & & set defined by $\left\{ xy \setst y\in\mathbb{Y} \right\}$\\
\bottomrule[\lightrulewidth]\bottomrule[\lightrulewidth]
\end{tabular*}
\end{center}
\section{\CG\ coefficients}\label{app:CG}
\begin{table}[h!]
\centering
\begin{tabular}{lccc} \toprule[\lightrulewidth]\toprule[\lightrulewidth]
& $(\Lambda,\Rho)=(\lambda-\shalf,\rho-\tfrac{A}{2})$ & $\quad$& $(\Lambda,\Rho)=(\lambda+\shalf,\rho+\tfrac{A}{2})$  \\ \midrule
$j=J-\half\qquad$ & $\ii A\frac{\sqrt{J-\lambda+\shalf}\sqrt{J-A\rho+\shalf}}{\sqrt{2J+1}\sqrt{\lambda+A\rho}}$ && $\frac{\sqrt{J+\lambda+\shalf}\sqrt{J+A\rho+\shalf}}{\sqrt{2J+1}\sqrt{\lambda+A\rho}}$ \\ \addlinespace[1em]
$j=J+\half$ & $
\frac{\sqrt{J+\lambda+\shalf}\sqrt{J+A\rho+\shalf}}{\sqrt{2J+1}\sqrt{\lambda+A\rho}}$ && $-\ii A\frac{\sqrt{J-\lambda+\shalf}\sqrt{J-A\rho+\shalf}}{\sqrt{2J+1}\sqrt{\lambda+A\rho}}$ \\ \bottomrule[\lightrulewidth]\bottomrule[\lightrulewidth]
\end{tabular}
\caption{\CG\ coefficients $\mathrm{B}\lbrace(\Lambda,\Rho)\,J|\gamma_A;(\lambda,\rho)\,j\rbrace$ for $\gamma=\half$.}
\label{tab:1/2}
\end{table}
Some notions about the \CG\ coefficients of the couplings $F^A_\gamma\otimes V_{\lambda,\rho}$ are presented here. In particular, explicit values for some \CG\ coefficients, namely those with $\gamma=\half$ are listed in Table~\ref{tab:1/2}; the normalisation is chosen so that
\begin{equation}
\mathrm{B}\lbrace(\Lambda,\Rho)\,J|\half_A;(\lambda,\rho)\,j\rbrace \equiv \mathrm{A}\lbrace\gamma_A;(\lambda,\rho)\,j|(\Lambda,\Rho)\,J\rbrace.
\end{equation}
Moreover, we give a proof of the following useful property:
\begin{appxprop}
Consider the product $F^A_\gamma\otimes V_{\lambda,\rho}$, with $\gamma\geq\shalf$ and $V_{\lambda,\rho}$ infinite-dimensional. If a \CG\ decomposition exists, when $J\geq |\lambda|+\gamma$ the \CG\ coefficients satisfy
\begin{equation*}
\frac{\mathrm{B}\lbrace(\Lambda+1,\Rho+A)\,J|\gamma_A;(\lambda,\rho)\,J-\gamma\rbrace}{\mathrm{B}\lbrace(\Lambda,\Rho)\,J|\gamma_A;(\lambda,\rho)\,J-\gamma\rbrace}\propto\frac{\sqrt{J+\Lambda+1}\sqrt{J+A\Rho+1}}{\sqrt{J-\Lambda}\sqrt{J-A\Rho}},\quad \forall (\Lambda,\Rho)\in\cC_J(\lambda,\rho,\gamma,A),
\end{equation*}
where the proportionality factor does not depend on $J$.
\end{appxprop}
\begin{proof}
Let $J\geq |\lambda|+\gamma$ and recall that in this case $\Omega_J(\lambda,\gamma)=\{J-\gamma,\dotsc,J+\gamma \}$, $\cC_J(\lambda,\rho,\gamma,A)$ does not depend on $J$, and one has
\begin{equation}
\label{eq:CGapp_vector}
\ket{\gamma_A,\gamma}\otimes\ket{(\lambda,\rho)\,J-\gamma,J-\gamma}\equiv\ket{(J-\gamma)\,J}=\sum_{(\Lambda,\Rho)\in\cC_J}\mathrm{B}\lbrace(\Lambda,\Rho)\,J|\gamma_A;(\lambda,\rho)\,J-\gamma\rbrace\ket{(\Lambda,\Rho)\,J}.
\end{equation}
Acting with $K_+$ on both sides of~\eqref{eq:CGapp_vector} we get respectively
\begin{equation}
\begin{split}
&-P^+_{\lambda,\rho}(J-\gamma)\sqrt{2J-2\gamma+1}\sqrt{2J-2\gamma+2}\ket{\gamma_A,\gamma}\otimes\ket{(\lambda,\rho)\,J+1-\gamma,J+1-\gamma}
\\
=&-P^+_{\lambda,\rho}(J-\gamma)\sqrt{2J-2\gamma+1}\sqrt{2J-2\gamma+2}\ket{(J+1-\gamma)\,J+1}
\\
=&-P^+_{\lambda,\rho}(J-\gamma)\sqrt{2J-2\gamma+1}\sqrt{2J-2\gamma+2}\sum_{(\Lambda,\Rho)\in\cC_{J}}\mathrm{B}\lbrace(\Lambda,\Rho)\,J+1|\gamma_A;(\lambda,\rho)\,J+1-\gamma\rbrace\ket{(\Lambda,\Rho)\,J+1}
\end{split}
\end{equation}
for the lhs and
\begin{equation}
-\sum_{(\Lambda,\Rho)\in\cC_{J}}\mathrm{B}\lbrace(\Lambda,\Rho)\,J|\gamma_A;(\lambda,\rho)\,J-\gamma\rbrace P^+_{\Lambda,\Rho}(J)\sqrt{2J+1}\sqrt{2J+2}\ket{(\Lambda,\Rho)\,J+1}
\end{equation}
for the rhs; equating both results we obtain, for each $(\Lambda,\Rho)\in\cC_{J}(\lambda,\rho,\gamma,A)$,
\begin{multline}
\label{eq:CGapp_equality}
\mathrm{B}\lbrace(\Lambda,\Rho)\,J|\gamma_A;(\lambda,\rho)\,J-\gamma\rbrace P^+_{\Lambda,\Rho}(J)\sqrt{2J+1}\sqrt{2J+2}=\\
\mathrm{B}\lbrace(\Lambda,\Rho)\,J+1|\gamma_A;(\lambda,\rho)\,J+1-\gamma\rbrace P^+_{\lambda,\rho}(J-\gamma)\sqrt{2J-2\gamma+1}\sqrt{2J-2\gamma+2}.
\end{multline}
Now let
\begin{equation}
f_J(\Lambda,\Rho):=\frac{\mathrm{B}\lbrace(\Lambda+1,\Rho+A)\,J|\gamma_A;(\lambda,\rho)\,J-\gamma\rbrace}{\mathrm{B}\lbrace(\Lambda,\Rho)\,J|\gamma_A;(\lambda,\rho)\,J-\gamma\rbrace},\quad (\Lambda,\Rho)\in\cC_{J}(\lambda,\rho,\gamma,A),
\end{equation}
where the numerator may vanish if $(\Lambda+1,\Rho+A)\not\in\cC_{J}(\lambda,\rho,\gamma,A)$; it follows from~\eqref{eq:CGapp_equality} that
\begin{equation}
f_{J+1}(\Lambda,\Rho)=
\frac{P^+_{\Lambda+1,\Rho+A}(J)}{P^+_{\Lambda,\Rho}(J)}f_J(\Lambda,\Rho)=
\frac{\sqrt{J+\Lambda+2}\sqrt{J+A\Rho+2}}{\sqrt{J-\Lambda+1}\sqrt{J-A\Rho+1}}\frac{\sqrt{J-\Lambda}\sqrt{J-A\Rho}}{\sqrt{J+\Lambda+1}\sqrt{J+A\Rho+1}}f_J(\Lambda,\Rho).
\end{equation}
One can check recursively that it must be, for each $n\in\N$,
\begin{equation}
f_{J+n}(\Lambda,\Rho)=\frac{\sqrt{J+n+\Lambda+1}\sqrt{J+n+A\Rho+1}}{\sqrt{J+n-\Lambda}\sqrt{J+n-A\Rho}}\frac{\sqrt{J-\Lambda}\sqrt{J-A\Rho}}{\sqrt{J+\Lambda+1}\sqrt{J+A\Rho+1}}f_J(\Lambda,\Rho);
\end{equation}
the solution of this \emph{recurrence relation} in $J$ is 
\begin{equation}
f_J(\Lambda,\Rho)\propto\frac{\sqrt{J+\Lambda+1}\sqrt{J+A\Rho+1}}{\sqrt{J-\Lambda}\sqrt{J-A\Rho}},
\end{equation}
where the proportionality constant is fixed by the normalisation of the \CG\ coefficients and does not depend on $J$.
\end{proof}

\section{Tridiagonal matrices}\label{app:tridiagonal}
\emph{Tridiagonal} matrices are square matrices whose only non-zero entries are on the main diagonal, the diagonal below it (\emph{subdiagonal}) and the diagonal above it (\emph{superdiagonal}). They can be visualised as 
\begin{equation}
A=
\begin{pmatrix}
b_1 & c_1\\
a_2 & b_2 & c_2\\
& \ddots & \ddots & \ddots\\
&& a_{n-1} &b_{n-1} & c_{n-1}\\
& & & a_n & b_n
\end{pmatrix},
\end{equation}
with the generic entry given by
\begin{equation}
\label{eq:tridiagonal}
A_{ij}=\delta_{i-1,j}\, a_i + \delta_{i,j}\,b_i + \delta_{i+1,j}\, c_i,
\end{equation}
where
\begin{equation}
a_1:=0 \quad \textnormal{and}\quad c_n:=0.
\end{equation}
A result holding for a certain class of tridiagonal matrices, originally presented in~\cite{wigner_eckart}, will be proved here.
\begin{appxprop}\label{prop:tridiagonal}
Let $A$ be a $n\times n$ tridiagonal matrix over a field $\K$.
If the superdiagonal (subdiagonal) entries of $A$ are all non-vanishing, its eigenspaces are all $1$-dimensional.
\end{appxprop}
\begin{proof}
Consider the case of non-zero superdiagonal entries.
Recall that, if $\lambda\in\K$ is an eigenvalue of $A$, the associated eigenspace is $\ker(A-\lambda\1)$, the vector space of solutions to the equation
\begin{equation}
Ax=\lambda x,\qquad x\in \K^n;
\end{equation}
with the notation introduced in~\eqref{eq:tridiagonal}, this is equivalent to the system of $n$ equations
\begin{equation}
\begin{cases}
\left(b_1-\lambda\right)x_1 + c_1\,x_2 = 0\\
a_i\, x_{i-1} + \left(b_i-\lambda\right)x_i + c_i\,x_{i+1}=0,\qquad i=2,\dotsc,n-1\\
a_n\,x_{n-1} + \left(b_n-\lambda\right)x_n=0.
\end{cases}
\end{equation}
If $x_1=0$ the first equation reduces to
\begin{equation}
c_1\,x_2=0,
\end{equation}
which implies $x_2$ is zero as well, since all the $c$'s are non-vanishing. In general, the $k$th equation will be
\begin{equation}
c_k\,x_{k+1}=0,
\end{equation}
i.e., the only solution with $x_1=0$ is the null vector.

Let then $x_1$ be an arbitrary non-zero value. Substituting each equation in the next one, the first $n-1$ equations reduce to a system of equations of the form
\begin{equation}
\label{eq:tridiagonal_solution}
c_i\,x_{i+1}=\alpha_{i+1}\,x_1,\qquad i=1,\dotsc,n-1,
\end{equation}
with each $\alpha$ depending solely on $\lambda$ and on the matrix entries.
These always have solution, since one can safely divide by the $c$'s;
as a consequence, the solution is \emph{completely} specified by the value of $x_1$, which can be factored out as a scalar coefficient.
The $n$th equation is automatically satisfied, as it was assumed that $\lambda$ is an eigenvalue.
By virtue of eqs. (\ref{eq:tridiagonal_solution}), all the non-zero solutions of the eigenvalue equation are proportional to each other, so that
\begin{equation}
\dim \ker(A-\lambda\1)=1.
\end{equation}

The proof for the case of non-zero subdiagonal entries prooceeds analogously.
\end{proof}
\bibliography{bibliography}
\end{document}